\documentclass[11pt,english]{smfart}

\usepackage{a4wide}

\usepackage[english]{babel}
\usepackage{amsmath}
\usepackage{amsfonts}
\usepackage{amssymb}
\usepackage{amsthm}
\usepackage{a4wide}
\usepackage{mathrsfs}
\usepackage[T1]{fontenc}
\usepackage{tikz}
\usepackage{shuffle}
\usepackage[all]{xy}
\usepackage{xcolor}
\usepackage[latin1]{inputenc}

\usepackage{nomencl}
\usepackage[intoc]{nomentbl}
\makenomenclature

\def\di{\displaystyle}
\newcommand{\R}{\mathbb{R}}
\newcommand\del[1]{}
\newcommand\think[1]{}
\newcommand\new[1]{}
\newcommand\zus[1]{}

\newcommand\comd[1]{} 
\newcommand\Redd[1]{} 

\def\bdm{\begin{displaymath}}
\def\edm{\end{displaymath}}
\def\bea{\begin{eqnarray}}
\def\eea{\end{eqnarray}}

\newtheorem{theorem}{Theorem}[section]
\newtheorem{lemma}[theorem]{Lemma}
\newtheorem{definition}[theorem]{Definition}

\newtheorem{corollary}[theorem]{Corollary}

\begin{document}

\title[A stochastic invariantization method]{A stochastic invariantization method for It\^o stochastic perturbations of differential equations}

\author{Jacky Cresson, Yasmina Kheloufi, Khadra Nachi}
\address{Laboratoire de Math\'ematiques Appliqu\'ees de Pau, 
 Universit\'e de Pau et des Pays de l'Adour, 
Avenue de l'Universit\'e, BP 1155,
64013 Pau Cedex, France
\&
Universit\'e d'Es-S\'enia, Oran -- D\'epartement de math\'ematiques, Alg\'erie
}
\email{jacky.cresson@univ-pau.fr}
\address{Universit\'e d'Es-S\'enia, Oran -- D\'epartement de math\'ematiques, Alg\'erie}
\email{kheloufiyasmina@ymail.com}

\begin{abstract}
In general, adding a stochastic perturbation to a differential equation possessing an invariant manifold destroys the invariance as far as the It\^o formalism is used. In this article, we propose an invariantization method for perturbations in the It\^o case which can be used to restore invariance. We then apply our results to develop a stochastic version of the Landau-Lifshitz equation. We discuss in particular previous results obtained by Etore and al. in \cite{Etore}.
\end{abstract}

\maketitle

\tableofcontents

\section{Introduction}

We consider a deterministic ordinary differential equation of the form
\begin{equation}
\label{ode}
\di\frac{dx}{dt}=f(x),\ \ x\in \mathbb{R}^n,\ n\in \mathbb{N}^* .
\end{equation}
A {\bf stochastic perturbation} is taken into account by adding a "noise" term to the classical deterministic equation as follows:
\begin{equation}
\di\frac{dx}{dt}=f(x) +"\mbox{\rm noise}" .
\end{equation}
and to replace the "noise" term by a stochastic one as
\begin{equation}
\label{sde}
dX_t = f(X_t ) dt + \sigma (X_t ,t) dW_t ,
\end{equation}
where $W_t$ is standard Wiener process. This procedure is for example well discussed in (\cite{oksendal2003stochastic}).\\ 

Of course, the main problem is in this case to {\bf find} the {\bf form of the stochastic perturbation}. We do not discuss this problem which is very complicated. We restrict our attention to the {\bf selection problem} which is concerned with the characterization of the set of {\bf admissible} stochastic models for a given phenomenon. By admissible we mean that the stochastic model satisfies some known constraints like positivity of some variables, conservation law, etc.  This {\bf selection} of a good candidate for a stochastic model of the phenomenon can be done in many ways. However, in our particular setting, dealing with the {\bf stochastic extension} of a known deterministic model, this selection is related to {\bf preserving} some specific {\bf constraints} of the phenomenon. For example, part of the Hodgkin-Huxley model describes the dynamical behavior of concentrations which are typically variables which belongs to the interval $[0,1]$. This property is independent of the particular {\bf dynamics} of the variables but is related to their {\bf intrinsic} nature. The same is true for the total energy of a mechanical system. This quantity must be preserved independently of the dynamics. We formulate the {\bf stochastic persistence problem} following our approach given in \cite{CrSz} in a different setting: \\

\noindent {\bf Stochastic persistence problem} : {\it Assume that a classical ODE of the form (\ref{ode}) satisfies a set of properties $\mathcal{P}$. Under which conditions a stochastic perturbation of the form (\ref{sde}) satisfies also properties $\mathcal{P}$ ?}\\

The previous problem lead to {\bf characterize} the set of $\sigma$ preserving the considered properties $\mathcal{P}$. classical properties are: Invariance of a given submanifold of $\mathbb{R}^n$, number of equilibrium points, stability properties of the equilibrium points, etc.\\

The literature on invariance of manifolds for stochastic differential equations is huge and most of the time abstract in particular for what concerns the stochastic analogue of the Nagumo-Brezis theorem. This explain perhaps why these results are not so well known in the applied community because the formulation of the conditions are not transparent for a given concrete system. In this article, we give a direct and simple derivation of a necessary and sufficient condition on the diffusion part in order that a submanifold globally defined as the preimage of a smooth function is preserved under a stochastic perturbation. The result depends drastically on the stochastic differential framework that one uses. In the Stratonovich case, the condition on the drift is the same as the one on the drift part. However, in the It\^o case, which covers most of the applications, the constraints for invariance are so strong that in many cases, one is enable to find a large class of admissible stochastic models. All these problems are discussed in Section \ref{sec-invariance}.\\

What to do if the framework for the model has to be the It\^o one ? An idea is then to developp a systematic and algorithmic {\bf invariantization method} in order to restore invariance in the It\^o case. This is done in Section \ref{sectioninvariantizationmethod} following an idea initiated by Etore and al \cite{Etore}. In order to illustrate our method, we apply it to an It\^o version of the Kubo stochastic Hamiltonian system and in Section \ref{landau} to construct a stochastic version of the {\bf Landau-Lifshitz equation}. We finally discuss the limitations and problems posed by the invariantization method.

\section{Reminder about It\^o stochastic differential equations}
\label{remindstoc}

In this article, we consider a parameterized differential equation of the form 

\begin{equation}
\label{DE}
dX_t = f (t,X_t,b)dt,\quad x\in \R^{n}
\tag{DE}
\end{equation}
where $b\in \R^{k}$ is a set of parameters, $f: \R^{n}\times \R^{k}\longrightarrow \R^{n}$ is a Lipschitz continuous function with respect to $x$ for all $b$.
We remind basic properties and definition of stochastic differential equations in the sense of It\^o. We refer to the book \cite{oksendal2003stochastic} for more details.\\

A {\it stochastic differential equation} is formally written (see \cite{oksendal2003stochastic},Chap.V) in differential form as  

\begin{equation}
\label{IE}
dX_t = f (t,X_t)dt+\sigma(t,X_t)dB_t ,
\tag{IE}
\end{equation}

which corresponds to the stochastic integral equation
\begin{equation}
X_t=X_0+\int_0^t f (s,X_s)\,ds+\int_0^t \sigma (s,X_s)\,dB_s ,
\end{equation}
where the second integral is an It\^o integral (see \cite{oksendal2003stochastic},Chap.III) and $B_t$ is the classical Brownian motion (see \cite{oksendal2003stochastic},Chap.II,p.7-8).\\

An important tool to study solutions to stochastic differential equations is the {\it multi-dimensional It\^o formula} (see \cite{oksendal2003stochastic},Chap.III,Theorem 4.6) which is stated as follows : \\

We denote a vector of It\^o processes by $\mathbf{X}_t^\mathsf{T} = (X_{t,1}, X_{t,2}, \ldots, X_{t,n})$ and we put $\mathbf{B}_t^\mathsf{T} = (B_{t,1}, B_{t,2}, \ldots, B_{t,l})$to be a $l$-dimensional Brownian motion (see \cite{karatzas},Definition 5.1,p.72),  $d\mathbf{B}_t^\mathsf{T} = (dB_{t,1}, dB_{t,2}, \ldots, dB_{t,l})$. We consider the multi-dimensional stochastic differential equation defined by (\ref{IE}). Let $F$ be a $\mathcal{C}^2(\mathbb{R}_+ \times \mathbb{R},\mathbb{R})$-function and $X_t$ a solution of the stochastic differential equation (\ref{IE}). We have 
\begin{equation}
dF(t,\mathbf{X}_t) = \frac{\partial F}{\partial t} dt + (\nabla_\mathbf{X}^{\mathsf T} F) d\mathbf{X}_t + \frac{1}{2} (d\mathbf{X}_t^\mathsf{T}) (\nabla_\mathbf{X}^2 F) d\mathbf{X}_t,
\end{equation}
where $\nabla_\mathbf{X} F = \partial F/\partial \mathbf{X}$ is the gradient of $F$ w.r.t. $X$, $\nabla_\mathbf{X}^2 F = \nabla_\mathbf{X}\nabla_\mathbf{X}^\mathsf{T} F$ is the Hessian matrix of $F$ w.r.t. $\mathbf{X}$, $\delta$ is the Kronecker symbol and the following rules of computation are used : $dt dt = 0$, $dt dB_{t,i}  = 0$, $dB_{t,i} dB_{t,j} = \delta_{ij} dt$.

\section{Stochastic invariance of submanifolds}
\label{sec-invariance}

In this section, we derive an {\bf invariance criterion} for a submanifold denoted by $M$ of codimension $1$ of $\mathbb{R}^{n}$ which correspond to the zero set of a given function $F:\mathbb{R}^{n}\rightarrow \mathbb{R} $ of class $C^{2},$ i.e.
\begin{equation}
\label{manifold}
M=\lbrace x\in \mathbb{R}^{n} \setminus F(x)=0 \rbrace ,
\end{equation}
under the flow of a stochastic differential equation in the It\^o sense. This result is by itself not new and many general results are known in particular a {\bf stochastic Naguno-Brezis} Theorem as proved by Aubin-Da Prato in \cite{aubin} or A. Milian in \cite{milian}. However, most of these results are difficult to read for a non-specialist in the field of stochastic calculus. The main interest of the following computations are precisely that our criterion can be {\bf easily derived} using {\bf basic results} in stochastic calculus. 

\subsection{Geometric definition of invariance}

We consider an ordinary differential equation of the form 
\begin{equation} 
\label{ODE}
\left\{
\begin{array}{lll}
 \dot {x}_t & = & f (t,x_t ), \\
x(0) & = & x_0
\end{array}
\right.
\tag{ODE}
\end{equation}
where $f:\mathbb{R}^{+} \times \mathbb{R}^{n}\longrightarrow \mathbb{R}$ is a function of class $C^{1}$ and $x_0 \in \R^n$ is the initial condition.

\begin{definition}
A given submanifold $M \subset \mathbb{R}^{n}$ is said to be invariant under the flow of the differential equation \eqref{ODE} if for all $ x_{0}\in M ,$ the maximal solution $ x_{t}(x_{0})$ starting in $ x_{0} $ when $ t=0 $ satisfied $x_{t}(x_{0})\in M$ for all $ t\in \mathbb{R}^+.$
\end{definition}

We denote by $T_{x} M$ the tangent plane of $M$ at $x$, we can write the invariance condition as follows 
\begin{equation}
f(t,x)\in T_{x}M,\quad \text{ for all } (t,x)\in \mathbb{R}^+ \times M.
\end{equation}
As $M$ is of codimension 1, for all $x\in M $ we can define the normal vector $N(x) $ to the tangent hyperplane $ T_{x}M $ in $ x ,$ such that 
$$ T_{x} M =\lbrace y \in \mathbb{R}^{d}, y.N(x)=0\rbrace,$$   
then the invariance condition can be written as 
\begin{equation}
\label{IF}
N(x)\cdot f(t,x)=0, \quad \text{for all } (t,x)\in \mathbb{R}^+ \times M.
\end{equation}
When $M$ is of the form \eqref{manifold}, the normal vector to $M$ at $x$ is equal to $\nabla F(x).$ Then the invariance condition reads as 
\begin{equation}
\label{IFF}
\nabla F(x) \cdot f(t,x)=0,\quad \text{for all } (t,x)\in \mathbb{R}^+ \times M.
\tag{IF}
\end{equation}

In the stochastic case, the trajectories are continuous but nowhere differentiable. As a consequence the previous geometric condition can not be used. In the following we discuss two natural generalization of the notion of invariance in the stochastic setting.

\subsection{Strong stochastic invariance}

Let us consider a stochastic differential equation of the form \eqref{IE}. The stochastic character of the flow allows us to defined two natural notions of invariance.
 
\begin{definition}[Strong persistence]
A submanifold $M$ is invariant in the strong sense for the stochastic system \eqref{IE} if for every initial data $x_0\in M$ almost surely, the corresponding solution $x(t),$ satisfies 
$$ \mathbb{P}\lbrace F\left( x(t)\right)  =0, t \in [ t_0, +\infty )\rbrace =1,$$
i.e., the solution almost surely attains values within the submanifold $M.$
\end{definition}

A direct computation gives the following criterion for stochastic invariance:

\begin{theorem}[It\^o 's strong invariance]
\label{invarianceito}
Let $M$ be a submanifold defined by a function $F$ invariant under the deterministic flow associated to \eqref{DE}, i.e., 
$$  \nabla F(x)\cdot f (t,x)=0, \text{ for all } x\in M, t\geq 0$$
The submanifold $ M $ is strongly invariant under the flow of the stochastic system \eqref{IE}, if and only if, 
$$  \nabla F(x)\cdot \sigma (t,x)=0, \text{ for all } x\in M, t\geq 0$$
and 
\begin{equation}
\label{sc}
\sum_{i,j}\frac{\partial^{2}F}{\partial x_{i}\partial x_{j}}(x_{t})\sum^{k}_{l=1}\sigma_{i,l}(t,x_{t})\sigma_{j,l}(t,x_{t}) =0.
\end{equation} 
\end{theorem}

\begin{proof}
The essential tool in this case is the It\^o formula that will help us to formulate the invariance condition. Indeed, a process $x_{t}$ leaves the submanifold $M$ invariant if and only if for all initial condition $x_{0}\in M$ a.s, the stochastic process associated to $x_{t}$ satisfies $F(x_{t})=0$ for all $t$ almost surely where it is defined.\\

The multidimensional It\^o formula reads as
$$d[F(x_{t})]=\nabla F(x_{t})dx_{t}+\frac{1}{2}\sum_{i,j}\frac{\partial^{2}F}{\partial x_{i}\partial x_{j}}(x_{t})dx_{i}(t)dx_{j}(t).$$
So we obtain
$$d[F(x_{t}]=\nabla F(x_{t})f(t,x_t)dt + \nabla F(x_{t})\sigma(t,x_{t})dW_t+\frac{1}{2}\sum_{i,j}\frac{\partial^{2}F}{\partial x_{i}\partial x_{j}}(x_{t})\sum^{k}_{l=1}\sigma_{i,l}(t,x_{t})\sigma_{j,l}(t,x_{t})dt.$$ 
The gradient of $F$  being always normal to the tangent space of $M$, we have $\nabla F(x_{t})\cdot f (t,x_{t})=0$ since the manifold $M$ is assumed to be invariant in the deterministic case. It remains
\begin{equation}
\label{ifF}
 d[F(x_{t}]=\nabla F(x_{t})\sigma(t,x_{t})dW_{t}+\frac{1}{2}\sum_{i,j}\frac{\partial^{2}F}{\partial x_{i}\partial x_{j}}(x_{t})\sum^{k}_{l=1}\sigma_{i,l}(t,x_{t})\sigma_{j,l}(t,x_{t})dt.
 \end{equation}
The only contribution to the stochastic part is given by $ \nabla F(x_{t})\sigma(t,x_{t}) $ and is equal to zero if and only if the perturbation $\sigma$ satisfies the invariance condition \eqref{IFF}. Then the previous expression reduces to:
\begin{equation}
\label{itoderiva}
d[F(x_{t}]= \frac{1}{2}\sum_{i,j}\frac{\partial^{2}F}{\partial x_{i}\partial x_{j}}(x_{t})\sum^{k}_{l=1}\sigma_{i,l}(t,x_{t})\sigma_{j,l}(t,x_{t})dt.
\end{equation}
that give us the third conditional.\\

If we assume that the stochastic perturbation take the simplest case, where $\sigma_{i,j}=\delta_{i}^{j},$ we get the condition 
$$ \sum_{i=1}^{d}[\frac{\partial^{2}F}{\partial x_{i}^2}(x_{t})\sigma_{i,i}(t,x_{t})]^{2}=0,\;\forall(t,x)\in \mathbb{R}^+\times M. $$
\end{proof}

The previous Theorem indicates that unless a very specific form for $\sigma$ and $F$, there is no hope to recover invariance of a given manifold using a direct stochastic perturbation of a deterministic equation in the It\^o case.\\

As an example, we can specialize this result in the case of the sphere $S^2$ which will be important to study the invariance property of Landau-Lifshitz equation.

\begin{corollary}
\label{spherei}
The sphere $ S$ is invariant under the flow of the stochastic system \eqref{IE} if and only if the stochastic perturbation is null on the sphere i.e.,
$$\sigma_{i,i}(t,x)=0, \quad i=1,...,d\text{ for all }\ t\in \mathbb{R}^+\text{ and }\ x\in S^{d-1} .$$
\end{corollary} 

\begin{proof}
The proof follows from the fact that $F(x)=\di\sum_{i=1}^d x_i^2$ so that condition \ref{sc} reduces to
\begin{equation}
\sum_{i=1}^{d}[\sigma_{i,i}(t,x_{t})]^{2}=0,\;\forall(t,x)\in \mathbb{R}^+\times S .
\end{equation}
This concludes the proof.
\end{proof}

As a consequence, trying to impose the invariance of $S^2$ in the It\^o case "kill" the perturbation that is intended to be produced by the diffusion term.

\section{The stochastic invariantization method} 
\label{sectioninvariantizationmethod}

In this section, we develop a procedure to restore invariance of manifold following a procedure initiated by Etor\'{e} and al \cite{Etore} in a particular case. The basic idea is that in some cases, it is possible to project a flow which does not leave the manifold invariant on the manifold.

\subsection{First idea: projection procedure}

Consider submanifolds of codimention 1 of $\R^n,$ that is defined by a homogeneous function of degree $q\in \mathbb{N} ; F: \R^n\longrightarrow \R$ of class $C^2$, i.e.,
$$M=\lbrace x\in \R^n / F(x)=1\rbrace \text{ and } F(\lambda x)=\lambda^q F(x), \text{ for all } x\in \R^n, \lambda \in \R^+.$$
Let us assume that the coefficients of the system \eqref{IE} satisfy the invariance condition \eqref{IFF}, i.e.,
 $$ \nabla F(x)\cdot f (t,x)=\nabla F(x)\cdot \sigma(t,x)=0, \text{ for all } x\in M, t\geq 0,$$
where $\sigma$ is a vector of $\R^n$ and $W_t$ is a scalaire Brownian motion, and assume that 

$$  \sum_{i,j}^{n} \frac{\partial^{2}F}{\partial x_{i} \partial x_{j}}(x_t)\sigma_{i}(t,x_{t})\sigma_{j}(t,x_{t})dt \neq 0. $$ 
 Then, by Theorem \eqref{invarianceito} we know that $M$ is not invariant under the flow of \eqref{IE}.\\
 
 A very simple way to construct an invariant stochastic process is to project on the manifold. In general, a projection on a manifold is difficult to compute. In our case, it reduces simply to consider the stochastic process 
\begin{equation}
y_t=\frac{x_t}{F(x_t)^{\frac{1}{q}}} ,
\end{equation}
called the {\bf ''projected'' process} associated to the stochastic process $x_t$ and the function $F$.\\

Although simple, the previous method is in general not interesting. Indeed, the projected process satisfies in general a very complicated equation.
 
\begin{theorem}
Let assume that $x_t$ is the solution of It\^{o} equation\eqref{IE} and $M$ is defined by $F$ of class $C^2$ as above. The projected process $y_t$ satisifies the equation 

$$ \begin{aligned}
dy_k=  & \left[ F(x)^{-\frac{1}{q}}f_k(t,x) -\frac{1}{2q}\frac{\partial F}{\partial x_k} (x) .F(x)^{-\frac{1+q}{q}}\sigma_k^2 (t,x) \right] dt  \\ 
 & + \frac{1}{2}\di\sum_{i,j}^n \frac{-1}{q} x_k \sigma_i (t,x)\sigma_j (t,x)\left[ F(x)^{-\frac{1+q}{q}}\frac{\partial^2 F}{\partial x_i \partial x_j}(x) -\frac{1+q}{q}F(x)^{-\frac{1+2q}{q}}\frac{\partial F}{\partial x_i}(x)\frac{\partial F}{\partial x_j} (x) \right] dt \\
 &  + \left[ F(x)^{-\frac{1}{q}}\sigma_k(t,x) \right]  dW_t;
 \quad \text{ for each } k=1,...,n.
 \end{aligned}
$$
\end{theorem}
 
Even for a simple manifold as a sphere, the corresponding equation does not simplify.
 
\begin{corollary}
Let assume that $x_t$ is the solution of It\^{o} equation\eqref{IE}. The projected process $y_t=\frac{x_t}{F(x_t)^{\frac{1}{2}}} $ on the sphere $S^{n-1}$ satisfies the equation 
\begin{equation}
\begin{aligned}
dy_k= & \left[ F(x)^{-\frac{1}{2}}f_k(t,x) -\frac{1}{2} x_k F(x)^{-\frac{3}{2}} \sigma_k^2(t,x)+ \frac{3}{2}  x_k  \sum_{i,j}^n \sigma_i (t,x)\sigma_j (t,x) F(x)^{-\frac{5}{2}}x_ix_j \right]dt \\
&+   \left[  F(x)^{-\frac{1}{2}}\sigma_k(t,x) \right]dW_t;  \quad \text{ for each } k=1,...,n.
\end{aligned}
\end{equation}
\end{corollary}
 
The previous expression has many problems:

\begin{itemize}
\item The resulting stochastic differential is far from being simple and can not in general be written only with respect to $y_t.$

\item The form of the deterministic part can not be seen as a perturbation of $f.$ This induces difficulties for the interpretation of the new equation.
\end{itemize}

In the following, we follow a different strategy initiated by P.Etore and al in \cite{Etore}.

\subsection{The invariantization method}

We first introduce the notion of invariatized process.

\begin{definition}[Invariantized process]
Let $x_t$ be a diffusion process defined by 
\begin{equation}
\label{eqx}
dx_t=f(t,x_t)dt+\sigma(t,x_t)ddW_t.
\end{equation}
The invariantized process associated to \eqref{eqx} and the submanifold $M$ defined by $F$ is defined by 
\begin{equation}
\left\lbrace \begin{array}{l}
dy_t=f(t,x_t)dt+ \sigma(t,x_t) dW_t\\
x_t=\frac{y_t}{\left( F(y_t)\right) ^{\frac{1}{q}}}\\
y_0=y\in M
\end{array}
\right.
\end{equation} 
\end{definition}

This terminology is justified by the fact that we have 
$$ F(x_t)=1, \text{ for all } t\geq 0.$$

The method associating to a given process and a submanifold $M$ its invariantized process is called the {\bf invariantization} method.\\

The main property of the invariantization method is that the stochastic differential equation satisfied by $x_t$ is simple in the contrary to the projection method.

\begin{theorem}[Invariantization]
\label{invariantizationmethod}
Assume that $F$ is homogeneous of degree $q$ and that $F(y_t)$ is a non random process, i.e. that there exists a function $h(t)$ such that 
\begin{equation}
dF(y_t ) = h(t) dt .
\end{equation}
Then, denoting by $H(t)$ the function defined by
\begin{equation}
H(t)=1+\di\int_0^t h(s)\, ds ,
\end{equation}
the invariantized stochastic process associated to $y_t$ and $F$ satisfies the stochastic differential equation
\begin{equation}
dx_t=\left[ -\frac{1}{q}\frac{ \dot{H}(t)}{H(t)}x_t+\frac{1}{\left( H(t)\right)^{\frac{1}{q}} }f(t,x_t)\right]  dt+\frac{1}{\left( H(t)\right)^{\frac{1}{q}} }\sigma(t,x_t)dW_t 
\end{equation}
where $\dot{H}$ denotes the derivative of $H$ with respect to $t$.
\end{theorem} 

\begin{proof}
This is a simple computation. As $F(y_t)=H(t)$, the stochastic process $x_t =y_t / H(t)^{1/q}$ satisfies 
\begin{equation}
dx_t = -\frac{1}{q} \di\frac{\dot{H} (t)}{H(t)^{1/q +1}} y_t +\di\frac{1}{H(t)^{1/q}} dy_t .
\end{equation} 
Replacing $y_t$ by $x_t H(t)$ and $dy_t$ by its expression, we obtain
$$\begin{array}{lll}
  d x_t & = & \left[ -\frac{1}{q}\frac{ \dot{H}(t)}{H(t)}x_t+\frac{1}{\left( H(t)\right)^{\frac{1}{q}} }f(t,x_t)\right]  dt\\
  & & +\frac{1}{\left( H(t)\right)^{\frac{1}{q}} }\sigma(t,x_t)dW_t .
\end{array} $$
This concludes the proof.
\end{proof}

\subsection{Example: It\^o version of the Kubo oscillator model}

Let us consider the Kubo oscillator (see for example \cite{mil}) in the It\^o case, which can be written as
 \begin{equation}
 \label{kuboi}
 dX_t=J_{a}X_t dt+J_{\sigma}X_t dW_t,
 \end{equation}
where $X=\begin{pmatrix} X_1\\ X_2 \end{pmatrix}\in \R^{2}, a, \sigma \in\R, W_t$ is a 1-dimensional Brownian motion and 
$$ J_{k}= \begin{pmatrix} 0 & -k\\ k & 0\end{pmatrix}, \;\forall k\in \R.$$
The Stratonovich version of the Kubo oscillator has any circle $X_1^2+ X_2^2=r^2,\; \forall r\in \R^+,$ invariant under the flow. However, the circles are not invariant under the flow of the It\^{o} version of the Kubo oscillator. Indeed, we have using the It\^{o} formula with $F(X_1,X_2)=X_1^2+ X_2^2,$ that 
\begin{equation}
dF((X_1,X_2)= \sigma^2( X_1^2+ X_2^2)dt.
\end{equation}
Assuming that $F(X_1,X_2)=r^2,\; r\neq 0$ is invariant under the flow gives 
\begin{equation}
dF((X_1,X_2)= \sigma^2 r^2 dt.
\end{equation}
As a consequence, the condition $dF=0$ is satisfied if and only if $\sigma=0$. This means that in the It\^{o} case, invariance can not be preserved while the flow is stochastic (i.e., $\sigma \neq 0).$\\

If we apply the last transformation such that $X_t=\frac{Y_t}{\vert Y_t\vert}$ we find that $X_t$ is the solution of the stochastic system
\begin{equation}
 dX_t=\tilde{J}_{a,\sigma,t} \; X_t dt+  \frac{1}{\sqrt{\sigma^2 t+1}}J_{\sigma}X_t dW_t,
 \end{equation}
 where
 $$ \tilde{J}_{a,\sigma,t}= \begin{pmatrix} -\frac{\sigma^2}{2(\sigma^2 t+1)} & - \frac{a}{\sqrt{\sigma^2 t+1}}\\ \frac{a}{\sqrt{\sigma^2 t+1}} & -\frac{\sigma^2}{2(\sigma^2 t+1)} \end{pmatrix},$$
 which preserve the invariance of $S$ under the flow of the deterministic equation 
 \begin{equation}
 dX_t=J_{a}X_t dt, \quad \text{ for all }a \in \R.
 \end{equation}

\section{Application: a stochastic Landau-Lifshitz equations}
\label{landau}

In this Section, we derive a stochastic Landau-Lifshitz equation in the It\^o setting. We first remind the construction of the classical Landau-Lifshitz equation and then its main properties. We then review classical stochastic approach used by different authors and the difficulties associated with these models. 

\subsection{The Landau-Lifshitz equation}

The Landau-Lifshitz equation is a generalization of the classical Larmor equation. The Larmor equation is {\bf conservative}. However, {\bf dissipative processes} take place within dynamic {\bf magnetization processes}. The microscopic nature of this dissipation is still not clear and is currently the focus of considerable research \cite{arrott,bertotti2}. The approach followed by Landau and Lifshitz consists of introducing dissipation in a phenomenological way. They introduce an additional torque term that pushes magnetization in the direction of the effective field. The Landau-Lifshitz equation becomes
\begin{equation}
\label{LLg}
\frac{d\mu}{dt}=-\mu\wedge b -\alpha\mu \wedge(\mu \wedge b),
\tag{LLg}
\end{equation} 
where $\mu \in \R^3$ is the single magnetic moment, $\wedge$ is the vector cross product in $\mathbb{R}^3$, $b$ is the effective field and $\alpha >0$ is the damping effects.\\

As for the Larmor equation, this equation possess many particular properties which can be used to derive a stochastic analogue. We review some of them in the next Section.
 
\subsection{Properties of the Landau-Lifshitz equation}

In this Section, we give a {\bf self-contained} presentation of some classical features of the LL equation. Readers which are familiar with the LL equation can switch this Section.  

\subsubsection{Invariance}

The following result is fundamental is all the stochastic generalization of the LL equation.

\begin{lemma}
\label{invarianceLL}
Let $\mu (0) \in S^2$, then the solution $\mu_t$ satisfies for all $t\in \R$, $\parallel\mu_t \parallel =1$, i.e. the sphere $S^2$ is invariant under the flow of the LL equation.
\end{lemma}

We give the proof for the convenience of the reader.

\begin{proof}
Let $\mu_t$ be a solution of the LL equation. We have 
\begin{equation}
\left .
\begin{array}{lll}
\di\frac{d}{dt} \left [ \mu_t . \mu_t \right ] & = & 2 \mu_t .\di \frac{d\mu_t }{dt} ,\\
 & = & \mu_t . \left [ -\mu_t \wedge b -\alpha\mu_t \wedge(\mu_t \wedge b) \right ] .
\end{array}
\right .
\end{equation}
By definition of the wedge product, the vectors $\mu_t\wedge b$ and  $\alpha\mu_t \wedge(\mu_t \wedge b)$ are orthogonal to $\mu_t$ so that 
\begin{equation} 
\di\frac{d}{dt} \left [ \mu_t . \mu_t \right ] =0 .
\end{equation}
As a consequence, using the fact that $\mu_0 \in S^2$, we deduce that 
\begin{equation}
\parallel \mu_t \parallel =\parallel \mu_0 \parallel =1 ,
\end{equation}
which concludes the proof.
\end{proof}

As a consequence, a solution starting on the sphere $S^2$ will remains always on it. The sphere being a two dimensional compact manifold, we can use classical result to deduce the asymptotic behavior of the solutions. But first, let us compute the equilibrium points.

\subsubsection{Equilibrium points}

The equilibrium points of the LL equation are easily obtained.

\begin{lemma}
The LL equation possesses as equilibrium points $b/\parallel b\parallel$ and $-b/\parallel b\parallel$. 
\end{lemma}

We give the proof for the convenience of the reader.

\begin{proof}
An equilibrium point $\mu \in \R^3$ satisfies 
\begin{equation}
-\mu\wedge b -\alpha\mu \wedge(\mu \wedge b)=0,
\end{equation}
which gives 
\begin{equation}
-\mu\wedge b =\alpha\mu \wedge(\mu \wedge b) .
\end{equation}
The vector $\mu$ must be orthogonal to $\mu \wedge b$ and at the same time equal to $-\mu \wedge b$ up to a factor $\alpha >0$. As $\mu_0 \in S^2$, we have $\mu\not= 0$ and the only solution is 
\begin{equation}
\mu \wedge b =0 .
\end{equation}
We then obtain $\mu =\lambda b$, with $\lambda \in \R$. By Lemma \ref{invarianceLL}, we must have $\mu \in S^2$ so that $\lambda +\pm 1/\parallel b\parallel$. This concludes the proof.
\end{proof}

We see that the equilibrium point of the LL equation coincide with those of the Larmor equation.\\

The stability of the previous equilibrium point can be easily studied using the Lyapounov theory.

\subsection{Toward a stochastic Landau-Lifshitz equation}

In this Section, we discuss the usual way of deriving a stochastic analogue of the Landau-Lifshitz equation by considering an external perturbation of the effective magnetic field. We focus on the Stratonovich and the It\^o interpretation and we explain the strategy used in Etore and al. \cite{Etore} to bypass the obstruction that the It\^o version does not preserve the sphere $S^2$ using the invariantization method. 

\subsubsection{Classical approach to the stochastic Landau-Lifshitz equation}  

The main approach to deal with the stochastic behavior of the effective magnetic field is to assume that the effective field $b$ is subject to a stochastic perturbation $b+\epsilon "noise"$. Due to the linearity with respect to the parameter $b$, we obtain an equation of the form 
\begin{equation}
d \mu_t = [ -\mu_t \wedge b- \alpha \mu_{t}\wedge ( \mu_{t}\wedge b )] dt+ \varepsilon [ -\mu_{t}\wedge "noise"- \alpha \mu_t \wedge \mu_t \wedge "noise" ] .
\end{equation}
Interpreting the previous equation in the It\^o formalism of stochastic differential equation leads to the following stochastic model:
\begin{equation}
\label{ELL}
d \mu_t = [ -\mu_t \wedge b- \alpha \mu_{t}\wedge ( \mu_{t}\wedge b )] dt+ \varepsilon [ -\mu_{t}\wedge dW_{t}- \alpha \mu_t \wedge \mu_t \wedge dW_t ] ,
\tag{ELL}
\end{equation}
where the term $\sigma(t,x)= -x \wedge . -\alpha x\wedge(x\wedge .)$ can be written as
\begin{equation}
\label{sigmaetore}
\sigma(t,x)= 
\begin{pmatrix}
\alpha(x_3^2+x_2^2) &  x_3-\alpha x_1x_2 &-x_2 -\alpha x_3x_1\\
-x_3-\alpha x_1x_2 &\alpha (x_3^2+x_1^2 )&x_1 -\alpha x_3x_2 \\
x_2-\alpha x_1 x_3 & -x_1-\alpha x_3x_2 & \alpha(x_2^2+x_1^2) 
\end{pmatrix}
.
\end{equation}

Most of the authors use the Stratonovich formalism in order to give a sense to the previous equation. The main reason is that in this case, the invariance of $S^2$ is ensured. However, as pointed out by Etore and al. in \cite{Etore}, the Stratonovich version of the Landau-Lifshitz equation leads to several difficulties, such as the fact that the stability of the equilibrium points of the deterministic LL equation is lost. \\

The previous point has motivated the work \cite{Etore} in which the authors discuss the It\^o case. However, the It\^o approach lead to other difficulties. Details are given in the next Section.

\subsubsection{Stochastic It\^o perturbation of the Landau-Lifshitz equation}  

The It\^o version of the stochastic Landau-Lifchitz equation possesses many drawback and the main one follows directly from the invariance criterion that we derive in Corollary \ref{spherei}. Indeed, we have:

\begin{lemma}
The sphere $S^2$ is not invariant under the flow of the It\^o version of equation (\ref{ELL}).
\end{lemma}

\begin{proof}
By Corollary \ref{spherei}, the diffusion term must be zero on the sphere $S^2$ and all $t\in \R$. The condition on $\sigma$ on the diagonal terms implies that $\alpha=0$, i.e. that we can not have a dissipative term and we recover the Larmor equation in contradiction with our assumption that $\alpha \not= 0$. This concludes the proof.
\end{proof}

The previous result excludes the use of the It\^o formalism for a direct stochastic generalization of the Landau-Lifshitz equation. However, we can use the {\it invariantization method} exposed in Section \ref{sectioninvariantizationmethod} in order to obtain an It\^o model, related to the previous one, but which satisfies the invariance of the sphere $S^2$.

\subsubsection{Invariantization of the Landau-Lifshitz stochastic It\^o model}
\label{etoremodel}

Let us consider the stochastic LL equation (\ref{ELL}). The sphere $S^2$ is defined by the homogeneous function $F(x_1 ,x_2 ,x_3)=\di\sum_{i=1}^3 x_i^2$ of degree $2$. Let us consider the invariantized process defined by
\begin{equation}
\label{IELL}
\left \{ 
\begin{array}{lll}
d y_t & = & [ -\mu_t \wedge b- \alpha \mu_{t}\wedge ( \mu_{t}\wedge b )] dt+ \varepsilon [ -\mu_{t}\wedge dW_{t}- \alpha \mu_t \wedge \mu_t \wedge dW_t ] ,\\
\mu_t & = & \di\frac{y_t}{\parallel y_t\parallel} .
\end{array}
\right .
\end{equation}

A simple computation gives (see also (\cite{Etore}, Proposition 1)):

\begin{lemma}
The process $F(y_t )$ is random. Precisely, we have $d F(y_t) = 2\epsilon^2 (\alpha^2 +1 ) dt$.
\end{lemma}

As a consequence, Theorem \ref{invariantizationmethod} applies and we have:

\begin{lemma}
The invariantized stochastic differential equation associated to $F(x)=\di\sum_{i=1}^3 x_i^2$ and the It\^o stochastic differential equation (\ref{ELL}) is given by 
\begin{equation}
\label{etoreequation}
dx_t=\left[ -\frac{1}{2}\frac{ 2\epsilon^2 (\alpha^2 +1)}{2\epsilon^2 (\alpha^2 +1)t+1}x_t+\frac{1}{\sqrt{2\epsilon^2 (\alpha^2 +1)t+1}}f(t,x_t)\right]  dt+\frac{1}{\sqrt{2\epsilon^2 (\alpha^2 +1)t+1} }\sigma(t,x_t)dW_t ,
\end{equation}
\end{lemma}

\begin{proof}
This is a simple computation.
\end{proof}

We recover the {\bf Etore and al. version of the Stochastic Landau-Lifshitz equation} proposed in \cite{Etore}.

\subsection{About equilibrium points}

A natural question about the previous invariantized model is to up to which extent it answers to the reasonable constraints one waits for a stochastic version of the Landau-Lifshitz equation. For example, if one is interested in preserving equilibrium points of the initial system, the model is not satisfying. Indeed, we have:

\begin{lemma}
The points $\mu =\pm b /\parallel b\parallel$ are not equilibrium points of equation (\ref{etoreequation}).
\end{lemma}

\begin{proof}
For $\mu=\pm b/\parallel b\parallel$, we have for all $v\in \R^3$ that  
\begin{equation}
\sigma(t,\pm b/\parallel b\parallel ).v= -\di\frac{b}{\parallel b\parallel} \wedge v -\pm \alpha \di\frac{b}{\parallel b\parallel} \wedge( \pm\di\frac{b}{\parallel b\parallel} \wedge v) .
\end{equation}
The second term is always zero but the first one is only zero when $v$ is collinear with $b$. However, as $v$ takes arbitrary values, we can not ensure this equality. As a consequence, the initial equilibrium points are destroyed under the stochastic perturbation.
\end{proof}

It must be noted that the previous problem can be easily solved by modifying a little bit the modeling of the stochastic behavior of the effective magnetic field. Indeed, let us consider instead of $dW_t$ the following vector 
\begin{equation}
b dW_t ,
\end{equation}
where $W_t$ is a one dimensional Brownian motion. This assumptions is equivalent to say that we consider only stochastic behavior in the direction of the initial field $b$. This is of course very particular, but in this case the new model preserves the equilibrium points of the initial system:

\begin{lemma}
Let us consider the modified Etore and al. stochastic Landau-Lifshitz equation defined by 
\begin{equation}
\label{modifiedetoreequation}
dx_t=\left[ -\frac{1}{2}\frac{ 2\epsilon^2 (\alpha^2 +1)}{2\epsilon^2 (\alpha^2 +1)t+1}x_t+\frac{1}{\sqrt{2\epsilon^2 (\alpha^2 +1)t+1}}f(t,x_t)\right]  dt+\frac{1}{\sqrt{2\epsilon^2 (\alpha^2 +1)t+1} }\sigma(t,x_t). \left [ bdW_t \right ] ,
\end{equation}
where $\sigma$ is defined by equation (\ref{sigmaetore}) and $W_t$ is a one dimensional Brownian motion. This equation possesses as equilibrium points $\pm b/\parallel b\parallel$.
\end{lemma}

\begin{proof}
This follows easily from the previous proof only saying that $v$ is always collinear to $b$.
\end{proof}
 
\section{The invariantization method as a stochastic perturbation}

Although the invariantization method leads to a simplest equation than the projection procedure, it is not very easy to understand the procedure as a stochastic perturbation of the deterministic model. In this Section, assuming that the diffusion is governed by a small parameter $0<\epsilon <<1$, we write the invariantized equation as a perturbation.

\subsection{Small perturbation and invariantization}

In the following, we use the notations of Section \ref{sectioninvariantizationmethod}. For $\sigma=0$, the invariantized process reduces to the deterministic equation. Let us assume that $\sigma$ is of the form 
\begin{equation}
\sigma (x)=\epsilon \sigma_0 (x),\ \ 0< \epsilon <<1 ,
\end{equation}
where $\sigma_0$ and $f$ satisfy the invariance conditions.\\

Using the It\^o formula, we obtain
\begin{equation}
dF(y_t) = \di\nabla F(y_t).dy_t +\di\frac{1}{2} \di{\partial^2}{\partial y^2} dy_t .dy_t .
\end{equation}
As $dy_t .dy_t = \epsilon^2 \sigma_0^2 (x_t) dt$, and $F$ and $\sigma_0$ satisfy the invariance relation, we finally obtain
\begin{equation}
dF(y_t) = \epsilon^2 \di\frac{1}{2} \di\frac{\partial F}{\partial y^2} \sigma_0^2 (x_t) dt.
\end{equation} 
Denoting by $gamma (t)$ the function
\begin{equation}
\gamma (t)=\di\frac{1}{2} \di\int_0^t \frac{\partial F}{\partial y^2} \sigma_0^2 (x_t) dt ,
\end{equation}
we then obtain using Theorem \ref{invariantizationmethod} a function $h_{\epsilon}$ of the form 
\begin{equation}
h_{\epsilon} (t)=\epsilon^2 \gamma (t) ,
\end{equation}
and a function $H_{\epsilon}$ of the form
\begin{equation}
H_{\epsilon} (t)=1+\epsilon^2 \delta (t) .
\end{equation}
As a consequence, we can develop the drift part with respect to $\epsilon$ and we obtain for the invariantized process an equation of the form 
\begin{equation}
\label{perturbationinvariantized}
dx_t =\left ( 
f +\di\frac{1}{q} (\epsilon^2 )^{1/q} \star +\dots 
\right )
dt
+
\left (
\epsilon \sigma_0 + \di\frac{1}{q} (\epsilon^2 )^{1/q} \tilde{\star} +\dots 
\right )
dW_t .
\end{equation}
We do not search for explicit expression of the perturbation terms $\star$ and $\tilde{\star}$.

\subsection{Limitation of the method}

From a modeling point of view, we believe that a stochastic model of a deterministic equation must satisfy the following constraints:
\begin{itemize}
\item First, the drift part must be relied early to the deterministic equation and moreover must be understandable as a perturbation of it, i.e. of the form 
\begin{equation}
\label{RP}
f(t,x(t))+P(t,x(t))
\end{equation}
where $P(t,x(t))$ is the perturbation term.

\item Second, the new equation must be easy to interpret and must keep a sense with respect to the field of applications.
\end{itemize}
\vskip 2mm
What can be said about the invariantization method ? \\

The first modeling constraint is then satisfied by our invariantization method. However, the perturbation term obtained in equation (\ref{perturbationinvariantized}) is very complicated and the role of each term in the dynamics is not easily recovered.

\section{Conclusion and perspectives}

For It\^o stochastic perturbation of ordinary differential equations, we have derived a general method allowing to preserve invariance of a given codimension one submanifold under the stochastic flow. This method has however some limitations and lead to difficulties in the interpretation of the resulting model from a perturbative point of view. A natural question is then to find other stochastization procedure which still use the It\^o formalism for stochastic differential equations but for which invariance can be ensured under reasonable constraints. We refer to \cite{ckp} where this problem is discussed in general in the framework of random ordinary differential equations and used to construct a new model of a stochastic Landau-Lifshitz equation.

\end{document}